\newcounter{myctr}
\def\myitem{\refstepcounter{myctr}\bibfont\noindent\ifnum\themyctr>9\else\phantom{0}\fi\hangindent17pt\themyctr.\enskip}

\documentclass{ws-ijqi}
\usepackage{hyperref}
\usepackage{braket}
\usepackage[super,sort,compress]{cite}

\begin{document}

\catchline{}{}{}{}{}

\title{Comment on ``Quantum key agreement protocol"}

\author{Nayana Das}

\address{Applied Statistics Unit, Indian Statistical Institute, India\\
dasnayana92@gmail.com}

\author{Ritajit Majumdar}

\address{Advanced Computing \& Microelectronics Unit, Indian Statistical Institute, India\\
majumdar.ritajit@gmail.com}

\maketitle

\begin{history}
\end{history}

\begin{abstract}
The first two-party Quantum Key Agreement (QKA) protocol, based on quantum teleportation, was proposed by Zhou et al. (Electronics Letters 40.18 (2004): 1149-1150). In this protocol, to obtain the key bit string, one of the parties uses a device to obtain the inner product of two quantum states, one being unknown, and the other one performs Bell measurement. However, in this article, we show that it is not possible to obtain a device that would output the inner product of two qubits even when only one of the qubits is unknown. This is so because the existence of such a device would imply perfectly distinguishing among four different states in a two-dimensional vector space. This is not permissible in quantum mechanics. Furthermore, we argue that the existence of such a device would also imply a violation of the ``No Signalling Theorem" as well. 
\end{abstract}



\markboth{N. Das \& R. Majumdar}
{Comment on ``Quantum key agreement protocol"}

\section{Introduction}
\label{sec:intro}
The emergence of quantum computing has left a significant impression on the field of cryptography. Since the time it was realized that quantum computers have the potential to decipher public key cryptosystems in polynomial time \cite{shor1999polynomial}, researchers have looked into protocols which make use of quantum mechanical laws in order to provide better security. The first ever quantum cryptographic protocol, to generate a shared secret key between two parties, was proposed by C. H. Bennett and G. Brassard in 1984 \cite{bennett1984quantum}. Quantum key distribution \cite{bennett1984quantum,ekert1991quantum,bennett1992experimental}, quantum secret sharing \cite{hillery1999quantum,cleve1999share,xiao2004efficient,gottesman2000theory}, quantum secure direct communication \cite{deng2004secure,wang2005quantum,das2020two,das2020improving,das2020cryptanalysis} and quantum key agreement (QKA) \cite{zhou2004quantum,tsai2009quantum,hsueh2004quantum,chong2010quantum,chong2011improvement,shi2013multi,liu2013multiparty,xu2014novel,shen2014two,huang2014cryptanalysis,huang2014quantum,he2015quantum,liu2016collusive,sun2016efficient,cai2017multi,he2019high,tang2020improvements} are some of the domains of quantum cryptography, in which the use of quantum mechanical laws has proven to be more efficient and secure than the existing protocols.

Key agreement is one of the basic requirements of cryptography, which allows two or more parties to agree on the same secret key by exchanging their information over public channels. The difference between key distribution and key agreement protocol is as follows: in key distribution protocol, one party can determine a key and distribute it to the other legitimate parties. But in a key agreement protocol, all the parties involved in the protocol contribute their information equally in order to generate a shared secret key. In 1976, Diffie and  Hellman first proposed a classical key agreement protocol \cite{diffie1976new}, whose security is based on the assumption that the discrete logarithm problem is computationally hard. But, in 1999, Shor proposed a quantum algorithm, which can solve the discrete logarithm problem in polynomial time \cite{shor1999polynomial}. Thus some key agreement protocol was in need, such that, the security of that protocol does not depend on the computation complexity of some mathematical problem.

The first two-party quantum key agreement protocol was proposed by Zhou et al.~\cite{zhou2004quantum}, who used quantum teleportation \cite{bennett1993teleporting} without the classical communication in order to generate the key bit string between the two parties. However, Tsai et al.~\cite{tsai2009quantum} pointed out that Zhou et al.'s protocol is not a fair key agreement protocol, since if one of the parties is malicious, then he can determine the key bits and distribute it to the other party without being detected. They also proposed an improved version of the previous one, but unfortunately, their protocol is still not a QKA as it does not satisfy the basic requirement of a key agreement protocol, i.e., the generated secret key is a sequence of random measurement results, which is not negotiated by any legitimate parties. In 2004, another two-party QKA protocol was proposed by Hsueh et al.~\cite{hsueh2004quantum} by using maximally entangled states. But later Chong et al. pointed out some security flaws in this protocol and showed that a malicious user can fully control the shared key alone~\cite{chong2011improvement}. They also proposed an improvement of this work. In 2010, Chong et al.~\cite{chong2010quantum} proposed a QKA protocol based on the BB84 protocol~\cite{bennett1984quantum} and the delayed measurement technique, where the sender encodes her key $K_A$ by preparing qubits in two different bases and sends those qubits to the receiver. Then the receiver encodes his key $K_B$ by applying some unitary operators on the received qubits and the final shared key becomes $K_{AB}=K_A \oplus K_B$. All the above-discussed protocols are two-party protocols. The first multi-party QKA protocol was proposed by Shi et al.~\cite{shi2013multi} by employing the Bell states and entanglement swapping~\cite{zukowski1993event}. But Liu et al. showed that it failed to be a QKA protocol as the shared key can be totally determined by a malicious participant alone, and they proposed the first secure multiparty QKA protocol with single particles~\cite{liu2013multiparty}.  Recently, other QKA protocols having different approaches are also explored~\cite{shen2014two,huang2014cryptanalysis,huang2014quantum,he2015quantum,liu2016collusive,sun2016efficient,cai2017multi,he2019high,tang2020improvements}.

In this article, we show that the first QKA protocol~\cite{zhou2004quantum}, from which some of the above-mentioned protocols follow, is not a valid protocol at all since it violates one fundamental law of quantum mechanics, namely - (i) the impossibility of distinguishing among a linearly dependent set of states \cite{nielsen2002quantum}, and one universal principle namely (ii) No signalling theorem \cite{kennedy1995empirical}.

Measurement of a quantum state is always associated with an orthonormal basis. Upon measurement, the quantum state collapses to one of the bases vector. Two quantum states $\ket{\psi}$ and $\ket{\phi}$ are said to be orthogonal if their inner product is zero, i.e. $\braket{\psi|\phi} = 0$. If $\ket{\psi}$ and $\ket{\phi}$ are orthogonal, and also the states are known to the observer, then he can prepare a measurement $M$ in the basis $\{\ket{\psi},\ket{\phi}\}$. When one of the two states $\ket{\psi}$ and $\ket{\phi}$ is provided to the observer, he can apply the measurement $M$, and the outcome will determine the state uniquely. However if the states $\ket{\psi}$ and $\ket{\phi}$ are non-orthogonal, they can not be distinguished reliably and if one consider a linearly dependent set of states (three or more states in two dimension), the probability of distinguishing among them is zero.

No signalling theorem states that it is not possible to transfer a signal instantaneously using quantum correlation. In other words, quantum mechanics does not have any contradiction with the theory of relativity, and the speed of transmission of information over a quantum channel is bounded by the speed of light. It can be noted that in quantum teleportation, one of the parties needs to send two classical bits of information to the other party. This sending of classical information provides a bound to the speed of information transmission by this protocol.

In the key agreement protocol proposed by Zhou et al. \cite{zhou2004quantum}, both of these fundamental laws are violated, which makes this protocol invalid. Moreover, this is not a valid key agreement protocol as one of the parties alone can determine the secret key without being noticed by the other party. Furthermore, if this party is dishonest, he can even make the other party decide on a key according to his own choice.

The rest of the paper is arranged as follows. In the next section, we briefly describe the QKA protocol of Zhou et al. \cite{zhou2004quantum}. In Section \ref{sec:problem}, we individually elaborate on the fallacies of this QKA protocol, and thus show why this protocol will fail to work. Finally a short conclusion is given in Section \ref{sec:conclusion}.

\section{\label{sec:review}Brief Review of Zhou et al.'s Protocol}
In this section, we briefly describe the ``Quantum key agreement" protocol \cite{zhou2004quantum}, proposed by Zhou et al. in 2004. This is a two-party QKA protocol, in which two parties, namely, Alice and Bob, want to establish a private key over a public channel by using the technique of quantum teleportation. The steps of the protocol are as follows:
\begin{enumerate}
	\item First Alice prepares a two-qubit maximally entangled state, also called the Singlet state, $\ket{\Psi^-}_{ab}=\frac{1}{\sqrt{2}}(\ket{0}_a\ket{1}_b-\ket{1}_a\ket{0}_b)$ in her lab, where $a$ and $b$ denote the first and second particles respectively. She keeps the first particle $a$ with her and sends the second particle $b$ to Bob. At the end of this step, Alice and Bob share a maximally entangled state between them.
	
	\item After receiving the particle $b$ from Alice, Bob prepares two identical single qubit states $\ket{\phi_c}$ and $\ket{\phi_d}$, where $\ket{\phi_c}=\ket{\phi_d} = \alpha\ket{0} + \beta\ket{1}$, and Bob knows the values of $\alpha$ and $\beta$. The composite state of three particles $a$, $b$ and $c$ is
	\begin{center}
		$\ket{\Psi}_{abc}=\ket{\Psi^-}_{ab}\ket{\phi_c}$,
	\end{center}
	which can be written as
	\begin{eqnarray*}
		\ket{\Psi}_{abc} &=& \frac{1}{2}[-(\alpha\ket{0}+ \beta\ket{1})_a \ket{\Psi^-}_{bc}\\
		&+& (\alpha\ket{0} - \beta\ket{1})_a \ket{\Psi^+}_{bc}\\
		&-& (\beta \ket{0}+ \alpha \ket{1})_a \ket{\Phi^-}_{bc}\\
		&+& (\beta \ket{0}- \alpha \ket{1})_a \ket{\Phi^+}_{bc}]
	\end{eqnarray*}
	where $\ket{\Phi^{\pm}}_{bc}=\frac{1}{\sqrt{2}}(\ket{0}_b\ket{0}_c \pm \ket{1}_b \ket{1}_c)$ and $\ket{\Psi^{\pm}}_{bc}=\frac{1}{\sqrt{2}}(\ket{0}_b\ket{1}_c \pm \ket{1}_b\ket{0}_c)$.
	
	\item Bob measures the particles $b$ and $c$ in Bell basis $\mathcal{B}=\{\ket{\Phi^{+}}, \ket{\Phi^{-}},\ket{\Psi^{+}},\ket{\Psi^{-}}\}$, and depending on the measurement outcome, the state of Alice's particle (let it be $\ket{\phi_a}$) collapses on one of these four states $(\alpha\ket{0} + \beta\ket{1})_a$, $(\alpha\ket{0} - \beta\ket{1})_a$, $(\beta \ket{0}+ \alpha \ket{1})_a$ and $(\beta \ket{0}- \alpha \ket{1})_a$ with equal probability. That is, the state of the particle $c$ is teleported to Alice. However, Bob does not disclose the outcome of his measurement, and hence neither Alice, nor any eavesdropper, knows the state of the qubit $\ket{\phi_a}$. At the same time Bob sends the qubit $\ket{\phi_d}$ to Alice through a private optical fibre.
	
	\item The measurement results $\ket{\Phi^{+}}, \ket{\Phi^{-}},\ket{\Psi^{+}}$ and $\ket{\Psi^{-}}$ denote the resultant key bits $11,10,01$ and $00$ respectively.
	
	\item \label{inner product} After receiving the particle $d$ from Bob, Alice calculates $\braket{\phi_a|\phi_d}$, i.e., the inner product of the two states of the particles $a$ and $d$. If there is no eavesdropper, then the possible results are $\alpha ^2 +\beta ^2=1, \alpha ^2 -\beta ^2, 2\alpha \beta, 0$ and these results correspond to the key bits $00, 01, 10, 11$ respectively. The relation between the outcomes and the key bits are shown in Table \ref{key table}. Bob chooses the values of $\alpha$ and $\beta$ such that $ \alpha ^2 -\beta ^2 \ne 2\alpha \beta \ne 0, 1$. Thus Alice and Bob both obtain their key bits.
\end{enumerate}

\begin{table}[htb]
	\centering
	\caption{Relation between the measurement outcomes and the values of the key bits}
	\begin{tabular}{|c|c|c|c|}
		\hline
		measurement  &  State of $\ket{\phi}_a$ & $\braket{\phi_a|\phi_d}$ & Key\\
		outcome of Bob & & & \\
		\hline
		$\ket{\Psi^-}_{bc}$ & $(\alpha\ket{0}+ \beta\ket{1})_a$ & $1$ & $00$\\
		\hline
		$\ket{\Psi^+}_{bc}$ & $(\alpha\ket{0}- \beta\ket{1})_a$ & $\alpha ^2 -\beta ^2$ & $01$\\
		\hline
		$\ket{\Phi^-}_{bc}$ & $(\beta\ket{0}+ \alpha\ket{1})_a$ & $2\alpha \beta$ & $10$\\
		\hline
		$\ket{\Phi^+}_{bc}$ & $(\beta\ket{0}- \alpha\ket{1})_a$ & $0$ & $11$\\
		\hline
	\end{tabular}
	\label{key table}
\end{table}

In 2009, Tsai and Hwang \cite{tsai2009quantum} pointed out that the above protocol is not a \emph{key agreement} protocol, since if Bob is malicious, then he can establish his preferable pre-defined key as the secret key, without being detected. To do this, Bob first measures the particles $b$ and $c$ in the Bell basis, and after observing the measurement result, he prepares the particle $d$ such that the result of the inner product implies his preferable key. 

\section{\label{sec:problem}Fallacies of this QKA Protocol}
In this section, we show that the QKA protocol by Zhou et al. violates two fundamental properties of quantum mechanics. We elaborate on each of them individually and then argue why the above protocol, let alone a valid \emph{key agreement}, is not a valid protocol as well.

\subsection{Implication on non-orthogonal state discrimination}
In the final stage of the protocol, Bob sends a qubit $\ket{\phi_d} = \alpha\ket{0} + \beta\ket{1}$ to Alice. He also teleports the state $\ket{\phi_c}$ to Alice, but does not disclose the outcome of his Bell measurement. In this stage, Alice has the qubit $\ket{\phi_a}$, which is in one of the following states
\begin{eqnarray*}
	\ket{\phi_a^1} &=& \ket{\phi_c} = \alpha\ket{0} + \beta\ket{1}\\
	\ket{\phi_a^2} &=& \sigma_x\ket{\phi_c} = \alpha\ket{1} + \beta\ket{0}\\
	\ket{\phi_a^3} &=& \sigma_z\ket{\phi_c} = \alpha\ket{0} - \beta\ket{1}\\
	\ket{\phi_a^4} &=& \sigma_z\sigma_x\ket{\phi_c} = \alpha\ket{1} - \beta\ket{0}\\
\end{eqnarray*}

Now, Alice uses her device to calculate the inner product $\braket{\phi_a^i|\phi_d}$, where $\alpha$ and $\beta$ are known and the outcome determines the state $\ket{\phi_a^i}$, as shown in Table \ref{key table}. We argue that, given the two qubits $\ket{\phi_a^i}$ and $\ket{\phi_d}$, it is not possible for Alice to physically calculate this inner product.

\begin{lemma}
	There exists no device which can calculate the inner product $\braket{\phi_a^i|\phi_d},\\ i= 1,2,3,4$.
\end{lemma}

\begin{proof}
	Let us assume that there is a device that can calculate the inner product mentioned above. The four possible qubit states $\ket{\phi_a^1},\ket{\phi_a^2},\ket{\phi_a^3}$ and $\ket{\phi_a^4}$ reside in a two dimensional Hilbert Space. Since four qubits can not be mutually orthogonal to each other in a two-dimensional vector space, therefore, the four qubit states cannot be mutually orthogonal to each other for any value of $\alpha,\beta$. As discussed in Sec.~\ref{sec:intro}, it is not possible to distinguish among the linearly dependent set of quantum states with a single copy. However, in this protocol, since the inner product is different for the different possible states of $\ket{\phi_a^i}$, if such a device exists it can perfectly distinguish between the four different states, which violates the postulate of quantum mechanics. Therefore, such a device cannot exist.
\end{proof}

This implies that it is not possible to physically calculate the inner product $\braket{\phi_a^i|\phi_d}$ given the two qubits.

\subsection{Violation of No signalling theorem}
In order to show  how the existence of such above mentioned device violates no-signalling condition, we consider a two qubits singlet state $\ket{\Psi^-}$ shared between Alice and Bob. This state has the property that it retains the same form in every spin direction \cite{nielsen2002quantum}. Therefore, for any two qubit states $\ket{\psi}, \ket{\bar{\psi}}$, such that $\braket{\psi|\bar{\psi}} = 0$, the singlet state can be expressed as
\begin{equation}
\ket{\Psi^-} = \frac{1}{\sqrt{2}}(\ket{\psi}\ket{\bar{\psi}} - \ket{\bar{\psi}}\ket{\psi})
\end{equation}

We assume $\alpha, \beta \in R$ and then from the previous subsection one easily note that $\braket{\phi_a^1|\phi_a^4} = 0$ and $\braket{\phi_a^2|\phi_a^3} = 0$. Therefore,
\begin{eqnarray}
\ket{\Psi^-} &=& \frac{1}{\sqrt{2}}(\ket{\phi_a^1}\ket{\phi_a^4} - \ket{\phi_a^4}\ket{\phi_a^1})\\
&=& \frac{1}{\sqrt{2}}(\ket{\phi_a^2}\ket{\phi_a^3} - \ket{\phi_a^3}\ket{\phi_a^2})
\end{eqnarray}

Now consider that Bob randomly decides to measure his qubit of the shared singlet state $\ket{\Psi^-}$ either in $M_1 = \{\ket{\phi_a^1},\ket{\phi_a^4}\}$ basis, or in $M_2 = \{\ket{\phi_a^2},\ket{\phi_a^3}\}$ basis. However, Bob does not disclose his choice of measurement basis, nor the outcome of his measurement.

If Bob measures his qubit in $M_1$ basis, then the density matrix of Alice's qubit is the following:
\begin{equation}
\rho_{14} = \frac{1}{2}(\ket{\phi_a^1}\bra{\phi_a^1} + \ket{\phi_a^4}\bra{\phi_a^4}).
\end{equation}

Whereas, if Bob measures in $M_2$ basis, then the density matrix of Alice's qubit is the following:
\begin{equation}
\rho_{23} = \frac{1}{2}(\ket{\phi_a^2}\bra{\phi_a^2} + \ket{\phi_a^3}\bra{\phi_a^3}).
\end{equation}

Since Alice and Bob can have an arbitrary spatial distance between them, if Alice can distinguish between $\rho_{14}$ and $\rho_{23}$, then she can determine the measurement basis chosen by Bob without any classical communication from Bob. This would violate the No signalling theorem.

However, in the protocol by Zhou et al \cite{zhou2004quantum}, Alice is equipped with a device, which can distinguish among the four different states $\ket{\phi_a^1},\ket{\phi_a^2},\ket{\phi_a^3}$ and $\ket{\phi_a^4}$. Therefore, with this device, Alice can also uniquely determine the state of her qubit after the measurement of Bob. If the state of Alice's qubit is (i) $\ket{\phi_a^1}$ or $\ket{\phi_a^4}$, then Bob must have measured in $M_1$ basis, and if the state is (ii) $\ket{\phi_a^2}$ or $\ket{\phi_a^3}$, then Bob must have measured in $M_2$ basis. Hence Alice can determine the measurement chosen by Bob without any classical communication by the later. This obviously violates the No signalling theorem. Even some weaker device that helps Alice distinguish between the two possible preparations ($\rho_{14}$ and $\rho_{23}$) of the same density matrix would also imply a violation of the no signalling principle \cite{nielsen2002quantum}.

\section{\label{sec:conclusion}Conclusion}
In this comment, we analyze the first two-party QKA protocol proposed by Zhou et al. in 2004 \cite{zhou2004quantum}. We observed that though the resources used in this protocol, namely one ebit of entanglement and physical transfer of one qubit, may be sufficient for key agreement for some properly chosen protocol, the suggested protocol is basically flawed. The protocol is not consistent with allowed physical operation in quantum mechanics. More importantly, the existence of such a protocol would imply violation of ``No signalling condition" obeyed by all physical theories including quantum mechanics.

\section*{Acknowledgement}
The authors would like to acknowledge Guruprasad Kar, Tamal Guha and Sutapa Saha of Physics and Applied Mathematics Unit, Indian Statistical Institute, for the stimulating discussions and their insightful comments.

\bibliographystyle{spphys}
\bibliography{QKA}

\end{document}